\documentclass[a4paper,12pt]{article}
\usepackage{amssymb,amsmath}
\usepackage[legacycolonsymbols]{mathtools}
\usepackage{bm}%
\usepackage{ascmac}
\usepackage{mathrsfs}
\usepackage{stmaryrd}
\usepackage{comment}
\usepackage[pdftex]{hyperref}
\hypersetup{
	colorlinks=true,
	citecolor=blue,
	linkcolor=blue,
	urlcolor=orange,
}
\usepackage{arydshln}
\usepackage{bigdelim}
\usepackage{physics}
\usepackage{cite}
\usepackage{comment}

\def\nn{\nonumber}
\def\Z{\mathbb{Z}}

\def\ua{\alpha}
\def\da{\dot{\alpha}}
\def\ub{\beta}
\def\db{\dot{\beta}}
\def\g{\mathfrak{g}}
\def\h{\mathfrak{h}}
\def\HLS{Haag-\L opusza\'nski-Sohnius }
\def\CM{Coleman-Mandula }
\def\ev{\text{ev.}}
\def\od{\text{od.}}

\usepackage{amsthm}
\theoremstyle{plain}
\newtheorem{thm}{Theorem}
\newtheorem*{thm*}{Theorem}

\theoremstyle{definition}
\newtheorem{dfn}{Definition}
\usepackage{version}	
\begin{document}

\title{Novel possible  symmetries of $S$-matrix generated by $\Z_2^n$-graded Lie superalgebras}

\author{Ren Ito and Akio Nago
	\\[15pt]
	 Department of Physics, Osaka Metropolitan University,
\\Sugimoto Campus, Osaka 558-8585, Japan}

\maketitle

\vfill
\begin{abstract}
In this paper, we explore the $\Z_2^n$-graded Lie (super)algebras as novel possible generators of symmetries of $S$-matrix. 
As the results, 
we demonstrate that 
a $\Z_2^n$-graded extension of the supersymmetric algebra can be a symmetry of $S$-matrix. Furthermore, it turns out that a $\Z_2^n$-graded Lie algebra appears as internal symmetries.
They are natural extensions of \CM theorem and \HLS theorem, which are the no-go theorems for generators of symmetries of $S$-matrix. 
\end{abstract}



\section{Introduction}
Since the introduction of $SU(4)$ symmetry in the classification of nuclei by Wigner in 1937, the properties and fundamental structure of particles have been described using group theory\cite{F.J.D, S.W}.
The quark model introduced in the Standard Model is a model using $SU(3)$ flavor symmetry proposed by Gell-Mann in 1964\cite{M.G}. 
It successfully describes the relationship among hadrons with the same spin.
Furthermore, if we consider the $SU(6)$ symmetry for quark flavor and spin, we can embed  hadrons with different spins in the same representation. 

However, such $SU(6)$ symmetry that unify and describe different spins are approximate symmetries that exist in the non-relativistic limit, and despite some success, all attempts to generalize them to relativistic quantum theory have failed\cite{S.C}. 
Indeed, it was shown by Coleman and Mandula in 1967 that such symmetries are forbidden in relativistic quantum field theory (QFT), known as the Coleman-Mandula theorem\cite{CM}.
The theorem says that all possible symmetries of $S$-matrix generated by Lie algebra are only Poincar\'e symmetry and internal symmetry such that its generators commute with Poincar\'e algebra.
This prevents us to unify Poincar\'e symmetry and internal symmetries into a single multiplet.
The reason for the strong restriction of this theorem is that the symmetries allowed in the S-matrix of QFT were searched only in the framework of Lie algebras. 
If we want to avoid the \CM theorem, it is natural to explore symmetries generated by  algebras beyond the Lie algebra.

In searching for symmetries that do not conflict with the Coleman-Mandula theorem, Haag, \L opsuza\'nski and Sohnius extended the algebra to Lie superalgebra, so that the symmetry generators are closed not only  for comutation relations but also for  anti-commutation relations\cite{HLS}.
This showed that supersymmetry is the only possible symmetry as an extension of Poincare symmetry, which is known as the Haag-\L opsuza\'nski-Sohnius theorem.
As its result, if we consider the Lie superalgebra as symmetry generators of the $S$-matrix, supersymmetry is allowed in QFT. 

In 1978,  further extensions of the Lie algebra were introduced by Rittenberg and Wyler \cite{RW1,RW2}(see also \cite{Ree,sch1}). This is called $\Z_2^n$-graded Lie (super)algebra, which is the extended algebra with the grading by abelian group $\Z_2^n:= \Z_2\times \Z_2 \times \cdots $(repeated $n$ times). It reduces to an ordinary Lie algebra when $n=0$,  and to a Lie superalgebra when $n=1$.
  
This type of algebra exhibits more complicated (anti-)commutation relations and possesses interesting properties in both physics and mathematics. In physics, symmetries generated by $\Z_2^n$-graded Lie superalgebras have been identified in various contexts, such as de Sitter SUGRA \cite{Vas}, nuclear quasi-spin \cite{jyw},  parastatistics \cite{tol},  non-relativistic  Dirac  equation \cite{AKTT1,AKTT2}, etc. 
Additionally, several $\Z_2^n$-graded extension of spacetime supersymmetries have  also been proposed \cite{zhe,LR,Toro1,Toro2,tol2,Bruce}. In particular, $\Z_2^n$-graded extension of super-Poincar\'e algebra introduced by Bruce \cite{Bruce} has some interesting properties and $\Z_2^n$-graded version of classical and quantum mechanics,  sigma model, etc. have been discussed based on the extension \cite{BruDup,AAD,AAd2,AKTcl,AKTqu,DoiAi1,DoiAi2,AiDoi,brusigma,bruSG,AIKT,Topp,Topp2,AiItoTa,AiItoTa2}. 

This indicates that $\Z_2^n$-graded Lie (super)algebras are not uncommon in physics. Instead, they provide new insights into the analysis of physical systems based on symmetries, which is one of the key motivations for studying their structure and representations. 
The structure and representation theories of $\Z_2^n$-graded Lie (super)algebras have been continuously explored since their introduction \cite{sch3,SchZha,Sil,ChSiVO,SigSil,PionSil,CART,MohSal,NAJS,NAPIJS,StoVDJ,Meyer,IsStvdJ,NAPSIJSinf,NAKA,KTclassification,Meyer2,NeliJoris,NeliJoris2,LuTan,FaFaJ,AiSe}.

However, despite the extensive studies carried out so far, the application of $\Z_2^n$-graded Lie (super)algebras to $(1+3)$-dimensional relativistic quantum field theory has not been considered at all. To this end, one of the most important issues is the $\Z_2^n$-graded Lie (super)algebra can be symmetries of $S$-matrix. 
If this is the case, it provides an extension of the \HLS theorem.
The purpose of this paper is to investigate the possible \( \mathbb{Z}_2^n \)-graded Lie (super)algebra symmetries of $S$-matrix in relativistic quantum field theory, following the approach of the \CM theorem and \HLS theorem. 
As the results, 
we demonstrate that 
a $\Z_2^n$-graded extension of the supersymmetric algebra can be a symmetry of $S$-matrix. Furthermore, it turns out that a $\Z_2^n$-graded Lie algebra appears as internal symmetries.

Our paper is organized as follows. In \S 2, we recall the definition of $\Z_2^2$-graded Lie superalgebra\cite{RW1,RW2}, and  briefly review \CM theorem and \HLS theorem\cite{CM,HLS} in the case of massive particles. 
In \S 3, we consider the novel generators of symmetries of $S$-matrix in the case of massive particles. We first analyze the scenario with $\Z_2^2$-grading, following the frameworks of \CM theorem and \HLS theorem. Subsequently, using the same method as for the $\Z_2^2$-grading, we will consider the possible generators with $\Z_2^n$-gradings. Finally we will discuss the relation between the generators of internal symmetries and $\Z_2^n$-graded Lie algebra.


\section{Preliminaries}
\subsection{$\Z_2^n$-graded Lie superalgebra}
We recall the definition of the $\Z_2^n$-graded LSA\cite{RW1,RW2}.
$\Z_2^n$ refers to the $n$-fold direct product of the abelian group $\mathbb{Z}_2=\{0,1\}$, namely $\Z_2^n=\mathbb{Z}_2\times \mathbb{Z}_2 \times \cdots$ (repeated $n$ times). 
\begin{dfn}
 Let $\g$ be a $\Z_2^n$-graded vector space  (over $\mathbb{R}$ or $ \mathbb{C}$), which is the direct sum of homogeneous vector subspaces labeled by an element of $\Z_2^n$:
\begin{align}
 \g = \bigoplus_{\vec{a}\in\Z_2^n}\g_{\vec{a}}.\label{Z2n-VS}
\end{align}
An element of $ \g_{\vec{a}}$ is said to have the $\Z_2^n$-grading $\vec{a}$. 
\end{dfn}
\begin{dfn}\label{defsalg}
We define the $\Z_2^n$-Lie bracket by
\begin{equation}
	\llbracket A, B \rrbracket = AB - (-1)^{\vec{a}\cdot\vec{b}} BA, \label{Z2n-LB}
\end{equation}
where $A \in \g_{\vec{a}}, \ B \in \g_{\vec{b}}$, and $ \vec{a}\cdot\vec{b} $ is the standard scalar product of $n$-dimensional vectors. 
A $\Z_2^n$-graded vector space is called a $\Z_2^2$-graded Lie superalgebra if 
$ \llbracket A,B \rrbracket \in \g_{\vec{a}+\vec{b}} $ and the Jacobi identity is satisfied:
\begin{align}
 \llbracket A, \llbracket B,C \rrbracket \rrbracket
= \llbracket \llbracket A,B \rrbracket, C \rrbracket
+ (-1)^{\vec{a}\cdot \vec{b}} \llbracket B, \llbracket A, C \rrbracket \rrbracket. \label{Z2n-JI}
\end{align}
\end{dfn}

The $\Z_2^n$-graded LSA can be regarded as an extended Lie algebra because the $\Z_2^n$-Lie bracket is a commutator (anti-commutator) for $ \vec{a}\cdot\vec{b} $ is even (odd). Actually when $n=0$, a $\Z_2^n$-graded LSA is equivalent to an ordinary Lie algebras. Furthermore, when $n=1$, it corresponds to the standard LSA, which consists of a bosonic sector ($0$-grading) and a fermionic sector ($1$-grading). It is known that for $n\ge 2$, a $\Z_2^n$-graded LSA also generates a symmetry. 
\subsection{Coleman-Mandula theorem}
In this subsection, we briefly review the Coleman-Mandula theorem, which is known as the no-go theorem for symmetries in QFT \cite{J.M, CM}.
The theorem states that
``The continious symmetries of the $S$-matrix in QFT can only be generated by a Lie algebra consisting of the translation generator $P_{\mu}$, the homogeneous Lorentz transformation generator $M_{\mu \nu}$, and the internal symmetry generator $B_l$''. 
Namely they discussed possible symmetries of the $S$-matrix generated by $\Z_2^n$-graded LSA's with $n=0$ and proved Poincar\'e symmetry and internal symmetry were only allowed in QFT. 
To consider the avoidance of this theorem, we look at the proof of this theorem\cite{J.M, S.W}.
They used the following resonable assumptions in their proof,
\begin{enumerate}
\item
The $S$-matrix is defined in relativistic local field theory, 
\item
The finite number of particle types with masses lighter than an arbitrary mass $M > 0$ is finite, 
\item
The amplitudes for elastic two-body scattering are analytic functions of the scattering angle at almost all energies and angles. 
\end{enumerate}
The symmetry generator here means a Hermitian operator on the Hilbert space which is commutative with the $S$-matrix and acts on multi particle states as well as single particle ones.
Note that we work in the momentum space so that any matrix element of the group generator is a distribution in momentum space.

First, we consider all generators which commute with $P_{\mu}$.
Any compact Lie algebra can be decomposed into a direct sum of a semi-simple part $\mathfrak{g}^\mathcal{S}$ and an Abelian part $\mathfrak{g}^\mathcal{A}$ (Levi-Decomposition). 
Since $P_\mu$ commute with itself, $P_{\mu}$ is in Abelian part: $P_{\mu} \in \mathfrak{g}^\mathcal{A}$. 
Thus, we consider the generators $B_{\alpha} \in \mathfrak{g}^\mathcal{S} $. 
We transform the commutation relation between $P_{\mu}$ and $B_\alpha$ by Lorentz transformation
\begin{align}
[U(\Lambda) B_\alpha U^{-1}(\Lambda),{\Lambda^\nu}_\mu P_\nu]=0
\end{align}
where $U(\Lambda)$ is the unitary operator on Hilbert space and represents the Lorentz transformation.
Since $ {\Lambda^\nu}_\mu $ is regular, $U(\Lambda) B_\alpha U^{-1}(\Lambda)$ must be a linear combination of $B_\alpha$:
\begin{align}
U(\Lambda) B_\alpha U^{-1}(\Lambda) = {D(\Lambda)^\beta}_\alpha B_\beta,
\end{align}
where $D(\Lambda)$ is the representation matrix of Lorentz group.
Since the metric of semi-simple part $\mathfrak{g}^\mathcal{S}$ is positive definite, $D(\Lambda)$ must be an identity matrix for $B_\alpha$ to satisfy the same algebraic relation after Lorentz transformation.
Noting that the Lorentz group is not compact, the symmetry generated by $B_\alpha \in \mathfrak{g}^\mathcal{S} $ is independent of Lorentz symmetry.
The $B_\alpha$ is independent of the momentum and acts as an identity matrix with respect to the spin, which is the generator of internal symmetry $B_l$.

Next, we consider all symmetry generators which do not commute with $P_\mu$.
All symmetry generators $G_{\ua}$ act on single particle states as follows
\begin{align} \label{alloperator}
G_\alpha \ket{n ,p}= \sum_{n^\prime} \int \dd^4 p^{\prime} \left( K_\alpha(p^{\prime}, p) \right)_{n^{\prime}, n}\ket{n^\prime ,p^\prime} ,
\end{align}
where $n^{\prime}, n$ denote the spin and particle type.
The kernel $K_\alpha (p^{\prime}, p)$ will be zero if both $p$ and $p^{\prime} $ are not on the mass hyperboloid. 
Furthermore, the kernel $K(p^{\prime}, p)$ is zero for all $p \neq p^\prime$:
\begin{align}
K_\alpha(p^{\prime}, p) = K_\alpha(p) \delta^4(p-p^{\prime}) .
\end{align}
From the second assumption, since the mass of the particle takes only discrete values, $G_\alpha$ commute with the mass square operator.
Thus, the commutation relation of $P_\mu$ and $G_\alpha$ satisfies
\begin{align}
[P^\mu,G_\alpha]=a_\alpha^{\mu\nu}P_\nu,
\end{align}
where $a_\alpha^{\mu \nu}$ is anti-symmetric with respect to $\mu,\nu$.
It follows that $G_\alpha$, which does not commute with $P_\mu$, is the Lorentz generator 
\begin{align}\label{Poincare}
[P_\mu, M_{\nu\rho}]=i(g_{\nu\mu}P_\rho - g_{\rho \mu}P_\nu).
\end{align}
Therefore, the symmetries of the $S$-matrix in QFT are only the Poincar\'e and internal symmetries.


\subsection{Haag-\L opusza\'nski-Sohnius theorem}
In this subsection, we briefly review the Haag-\L opusza\'nski-Sohnius theorem\cite{HLS}, which extends the Coleman-Mandula theorem.
In \eqref{alloperator}, $G_\alpha$ can be separated into two different types; bosonic type for same statistics $n^{\prime}, n$, or fermionic type for different statistics $n^{\prime}, n$.
All bosonic type generators are fully analyzed in the Coleman-Mandula theorem.
A complete basis of bosonic type generators are given by the $P_{\mu}$,$M_{\mu\nu}$ and  $B^l$.
On the other hand, the fermionic type generator is not mentioned in the Coleman-Mandula theorem because of the assumption that the all generators consist of Lie algebra.
They extended the Lie algebra to include anti-commutation relations in order to consider the fermionic type generators.
This theorem says that,
``Poincar\'e symmetry can only be extended to supersymmetry''

Now, let us see how this supersymmetry avoids the Coleman-Mandula theorem.
Note that $G_\alpha$ in \eqref{Poincare} can be bosonic or fermionic. 
However, the proof in the previous section concludes that the possible $G_\alpha$ is only \eqref{Poincare}.
Therefore, it is enough to consider fermionic generator $Q$ commute with $P_\mu$. 
Let $\mathfrak{c}$ denote the set of all generators commute with $P_\mu$.
Since $\mathfrak{c}$ is stable under the Lorentz transformations, it is a finite dimensional representation space of Lorentz groups, and all generators in $\mathfrak{c}$ can be decomposed into irreducible representation.
The irreducible representation of the Lorentz group is specified by two spinors, which we write as $(a,b)$-rep .
In general, because fermionic generators can have an odd number of spin indices, we denote the generator of  $(j,j^\prime)$-rep  and its conjugate $(j^\prime,j)$-rep  characterizing two independent spins $\ua,\db$ as $Q_{\ua_1\dots\ua_{2j}\db_1\dots\db_{2j'}}, \ \bar{Q}_{\ub_1\dots\ub_{2j^\prime}\da_1\dots\da_{2j}}$,
where $j,\ j^\prime$ is an integer or half-integer, and $2j$ undotted and $2j^\prime$ dotted indices are symmetric in each part.
Now consider the anti-commutation relation between $Q$ and $\bar{Q}$.
Since it is producted from two fermionic generators, it must be bosonic and belong to $(j+j^\prime, j+j^\prime)$-rep.
The bosonic generators are defined by the Coleman-Mandula theorem and are either $P_\mu,M_{\mu\nu},B_l$. The irreducible representation of the Lorenz group of each generator is as follows
\begin{align}
P_\mu \rightarrow \left(\frac12,\frac12 \right)\text{-rep},\quad M_{\mu\nu} \rightarrow \left(1,0\right)\text{-rep} +\left(0,1\right)\text{-rep},\quad B_l \rightarrow \left(0,0\right)\text{-rep}. \nn
\end{align}
It follows that $j+j^\prime \leq \frac{1}{2}$, which means $Q \in \mathfrak{c}$ is only rank-1 spinor charge. 
Therefore, the anti-commutator of $Q$ with $\bar{Q}$ becomes $P_{\mu}$
\footnote{
Our notations are as follows. Let $g_{\mu\nu}=\mathrm{diag}(+,-,-,-)$ and $(\sigma^\mu)_{\ua\da}=(\sigma^0,\sigma^1,\sigma^2,\sigma^3)_{\ua\da}$, where
\begin{align*}
 \sigma^0=\smqty(\pmat{0}),\qquad
\sigma^1= \smqty(\pmat{1})
,\qquad
\sigma^2= \smqty(\pmat{2})
,\qquad
\sigma^3= \smqty(\pmat{3}).
\end{align*}
Define the complex conjugate for any matrices as $(M_{\ua\ub})^\ast=:\bar M_{\da\db}$ and the Hermitian conjugate as $(M_{\ua\ub})^\dag =:\bar M_{\db\da}$. Then we obtain $\{(\sigma^\mu)_{\ua\db}\}^\dag= (\bar \sigma^\mu)_{\ub\da}=(\sigma^\mu)_{\ub\da}$ and $(\bar \sigma^\mu)^{\da\ua}:=\varepsilon^{\da\db}\varepsilon^{\ua\ub}\sigma^\mu_{\ub\db}=(\sigma^0,-\sigma^i)$, where $\varepsilon_{\ua\ub}$ is the antisymmetric tensor characterized by $\varepsilon^{12}=\varepsilon_{21}=1$(the same for dotted indices). 
}
\begin{align}
\{ Q_{\ua}, \bar{Q}_{\db} \}= \sigma^\mu_{\ua\db}P_{\mu}.
\end{align}
$Q_\alpha, \bar Q_{\da}$ are the supersymmetry generators.
The other relations are obtained in the same way
\begin{align}
[Q_\ua, M_{\mu \nu}]= {(\sigma_{\mu\nu})_\alpha}^{\ua_1} Q_{\ua_1}, \quad
\{Q_\ua,Q_\ub\}=0.
\end{align}

If we consider that there are several charges $Q_{\ua\, L}, \bar{Q}_{\da\, M}\ (L,M=1,2,\dots )$, anti-commutation relations of $Q_{\ua, L}$ is non-zero:
\begin{align}
\{ Q_{\ua\, K}, Q_{\ub\, L}\} 
&= \epsilon_{\ua \ub} \sum_l a_{l,KL}B_{l} \eqqcolon \epsilon_{\ua \ub} Z_{KL},\nn\\
\{\bar Q_{\da\, K}, \bar Q_{\db\, L}\} 
&= \varepsilon_{\da \db} \sum_l \bar a_{l,KL}B_{l} \eqqcolon \varepsilon_{\da \db} \bar Z_{KL},
\end{align}
where $a^l_{LM}$ is antisymmetric with respect to $L,M$.
From the Jocobi identity, we can get
\begin{align}\label{central}
[ Z_{KL}, Z_{MN}] = [ Z_{KL}, \bar Z_{MN}] = 0.
\end{align}
Since $Z_{KL}$ commutes with itself, $Z_{KL}$ is in the Abelian part $\mathfrak{g}^{\mathcal{A}}$ of the Lie algebra and commutative with $B_l \in \mathfrak{g}^\mathcal{S}$.
Therefore, $Z_{KL}$ commutes with all elements of the algebra. 
It is called the central charge.
The LSA $\langle P_{\mu}, M_{\mu\nu}, Z_{KL}, \bar Z_{KL}, Q_{\ua\, L}, \bar{Q}_{\da\, M}\rangle$ is called the superPoincar\'e algebra, and the LSA $\langle P_{\mu}, Z_{KL}, \bar Z_{KL}, Q_{\ua\, L}, \bar{Q}_{\da\, M}\rangle$ is called the supersymmetric algebra.

\section{Novel possible generators of the $S$-matrix}
\setcounter{equation}{0}
\subsection{Possible $\Z_2^2$-graded generators of symmetries}
In this section, we search for novel possible generators of symmetries of $S$-matrix within the framework of the $\Z_2^2$-graded LSA.
From \eqref{Z2n-VS} $\Z_2^2$-graded LSA $\g$ has four types of homogeneous vector subspaces:
\begin{align}
 \g=\g_{(0,0)}\oplus\g_{(1,1)}\oplus\g_{(0,1)}\oplus\g_{(1,0)},
\end{align}
and, from \eqref{Z2n-LB}, $\Z_2^2$-Lie brackets between each subspace are given, in terms of (anti-)commutator as in Table \ref{Z22-LB}.
\begin{table}[h]
 \begin{align*}
  \begin{array}{c|cccc}
   & \g_{(0,0)}& \g_{(1,1)}& \g_{(0,1)}& \g_{(1,0)}\\ \hline
   \g_{(0,0)}& [,]& [,]& [,]& [,]\\
   \g_{(1,1)}& [,]& [,]& \{,\}& \{,\}\\ 
   \g_{(0,1)}& [,]& \{,\}& \{,\}& [,]\\
   \g_{(1,0)}& [,]& \{,\}& [,]& \{,\}\\
  \end{array}
 \end{align*}
\caption{$\Z_2^2$-Lie brackets between each subspace}
\label{Z22-LB}
\end{table}

\noindent Thus, the generator $G_{\ua}$ in \eqref{alloperator} is categorized into two types of bosonic generators in $\g_{(0,0)}$ and $\g_{(1,1)}$, and two types of fermionic generator in $\g_{(0,1)}$ and $\g_{(1,0)}$. The $(0,0)$-graded generators are (ordinary) bosonic because the algebraic relations between them and all generators are given by the commutator. 
For the $(0,1)$ and the $(1,0)$-graded generators, the relations among generators with the same grading are given by the anti-commutators, while those between generators with different gradings are given by commutators. They are called commuting fermionic generators.
The (1,1)-graded generators are referred to as exotic bosonic ones because the relations among them are given by commutators, whereas their relations with fermionic generators are given by anti-commutators. 
Note that there are four types of symmetry generators of $S$-matrix.
This means there are four types of particles in QFT within this framework. 

We focus on the following three subalgebras:
\begin{align*}
 \h_1=\g_{00}\oplus\g_{11},\qquad  \h_2=\g_{00}\oplus\g_{01},\qquad \h_3=\g_{00}\oplus\g_{10}.
\end{align*}
As indicated in Table \ref{Z22-LB}, $\h_1$ is an ordinary Lie algebra, while $\h_2$ and $\h_3$ are nothing but LSA's. 
This shows that we can obtain generators by applying the \CM theorem to $\h_1$ and the \HLS theorem to $\h_2$ and $\h_3$. 
Therefore, in $\h_1$, there are only the translation generators $P_\mu$, Lorentz generator $M_{\mu\nu}$ and two types of generators of internal symmetries: $B_l^{00}\in \g_{(0,0)}$, $B_l^{11}\in\g_{(1,1)}$. 
$P_\mu$ and $M_{\mu\nu}$, of course, commute with $B^{00}_l$ and $B^{11}_l$.
The algebraic relations between $B$'s are given by
\begin{alignat}{2}
 [B_l^{00},B_m^{00}]&=i \sum_k f_{lmk} B_k^{00},\qquad &
 [B_l^{11},B_m^{11}]&=i \sum_k g_{lmk} B_k^{00},\nn\\
 [B_l^{00},B_m^{11}]&=i \sum_k h_{lmk} B_k^{11},
\end{alignat}
where $f_{lmk}$, $g_{lmk}$ must be antisymmetric $f_{lmk}=-f_{mlk}$, $g_{lmk}=-g_{mlk}$, but $h_{lmk}$ has no conditions because $B^{00}_l$ has the different grading from $B^{11}_l$.

On the other hand, $\h_2$ is identical to the LSA introduced in \S 2.3.
Denoting the $(0,1)$-graded supercharges by $Q_{\ua\,K}^{01}$, $\bar Q_{\da\,K}^{01}\in \g_{(0,1)}$,
the algebraic relations between $P_\mu$, $B^{00}_l$ and $Q^{01}$'s are given by
\begin{alignat}{2}
 \{Q^{01}_{\ua\,K}, \bar Q^{01}_{\da\,L}\} &=\sigma^{\mu}_{\ua\da}P_\mu\delta_{KL},\qquad&
 [Q^{01}_{\ua\,K},B_l^{00}]&=\sum_L s_{l,KL}Q^{01}_{\ua\,L}, \nn\\
\{Q^{01}_{\ua\,K}, Q^{01}_{\ub\,L}\} &= \varepsilon_{\ua\ub} \sum_l a_{l,KL} B^{00}_l,\qquad&
\{\bar Q^{01}_{\da\,K}, \bar Q^{01}_{\db\,L}\} &= \varepsilon_{\da\db} \sum_l \bar a_{l,KL} B^{00}_l, 
\end{alignat}
where $a_{l,KL}=-a_{l,LK}$. Introducing the generator $Z^{(1)}_{KL}:=\sum_l a_{l,KL} B^{00}_l\in\g_{(0,0)}$, it becomes the central charge in $\h_2$, which means
\begin{alignat}{2}
\{Q^{01}_{\ua\,K}, Q^{01}_{\ub\,L}\} &= \varepsilon_{\ua\ub} Z^{(1)}_{KL},\qquad&
 [Z^{(1)}_{KL},\,\circ\,]&=0,
\nn\\
\{\bar Q^{01}_{\da\,K}, \bar Q^{01}_{\db\,L}\} &= \varepsilon_{\da\db} \bar Z^{(1)}_{KL},\qquad&
 [\bar Z^{(1)}_{KL},\,\circ\,]&=0.
\end{alignat}
Similarly, $\h_3$ is also identical to the LSA introduced in \S 2.3.
Denoting the $(1,0)$-graded supercharges by $Q_{\ua\,K}^{10}$, $\bar Q_{\da\,K}^{10}\in \g_{(1,0)}$ in $\h_3$, the algebraic relations between $P_\mu$ and $Q^{10}$'s are given by
\begin{alignat}{2}
 \{Q^{10}_{\ua\,K}, \bar Q^{10}_{\da\,L}\} &=\sigma^{\mu}_{\ua\da}P_\mu\delta_{KL},\qquad&
 [Q^{10}_{\ua\,K},B_l^{00}]&=\sum_L t_{l,KL}Q^{10}_{\ua\,L},\nn\\
\{Q^{10}_{\ua\,K}, Q^{10}_{\ub\,L}\} &= \varepsilon_{\ua\ub} \sum_l b_{l,KL} B^{00}_l,\qquad&
\{\bar Q^{10}_{\da\,K}, \bar Q^{10}_{\db\,L}\} &= \varepsilon_{\da\db} \sum_l \bar b_{l,KL} B^{00}_l, 
\end{alignat}
where $b_{l,KL}=-b_{l,LK}$. Introducing the generator $Z^{(2)}_{KL}:=\sum_l b_{l,KL} B^{00}_l \in\g_{(0,0)}$, it also becomes the central charge in $\h_3$, which means
\begin{alignat}{2}
\{Q^{10}_{\ua\,K}, Q^{10}_{\ub\,L}\} &= \varepsilon_{\ua\ub} Z^{(2)}_{KL},\qquad&
 [Z^{(2)}_{KL},\,\circ\,]&=0,
\nn\\
\{\bar Q^{10}_{\da\,K}, \bar Q^{10}_{\db\,L}\} &= \varepsilon_{\da\db} \bar Z^{(2)}_{KL},\qquad&
 [\bar Z^{(2)}_{KL},\,\circ\,]&=0.
\end{alignat}


We need to determine the other algebraic relations. Taking into account that $Q^{01}_{\ua\,L}$ and $Q^{10}_{\ua\,L}$ are a $(\frac{1}{2},0)$-rep, $[Q^{01}_{\ua\,K},Q^{10}_{\ub\,L}]\in \g_{(1,1)}$ will be a $(1,0)$-rep or $(0,0)$-rep, and  $[Q^{01}_{\ua\,K},\bar Q^{10}_{\da\,L}]\in \g_{(1,1)}$ will be a $(\frac{1}{2},\frac{1}{2})$-rep. Moreover, because $B^{11}_l \in \g_{11}$ is $(0,0)$-rep, we can determine the following relations: 
\begin{alignat}{2}
 [Q^{01}_{\ua\,K},Q^{10}_{\ub\,L}]&=\sum_l \varepsilon_{\ua\ub} c_{l,KL} B_l^{11},\qquad&
 [\bar Q^{01}_{\da\,K},\bar Q^{10}_{\db\,L}]&=- \sum_l \varepsilon_{\da\db} \bar c_{l,KL} B_l^{11}
,\nn\\
 \{Q^{01}_{\ua\,K},B_l^{11}\}&=\sum_L u_{l,KL}Q^{10}_{\ua\,L},\qquad &
 \{Q^{10}_{\ua\,K},B_l^{11}\}&=\sum_L v_{l,KL}Q^{01}_{\ua\,L}
,\nn\\
 [Q^{01}_{\ua\,K},\bar Q^{10}_{\da\,L}]&=0.
\end{alignat}

We compute $\Z_2^2$-Jacobi identities \eqref{Z2n-JI} to consider the conditions for structure constants such that $\g$ becomes $\Z_2^2$-graded LSA. 
We obtain the following conditions
from $(A,B,C)=(Q,\bar Q, B)$:
\begin{alignat}{3}
 s_{l,LK}=&\bar s_{l,KL},\qquad&  t_{l,LK}=&\bar t_{l,KL},\qquad&  u_{l,LK}=&\bar v_{l,KL},
\end{alignat}
from $(A,B,C)=(Q,Q,\bar Q)$:
\begin{alignat}{2}
\sum_l a_{l,KL}\bar s_{l,KL}=&0
,\qquad& 
\sum_l b_{l,KL}\bar t_{l,KL}=&0
,\nn\\ 
\sum_l c_{l,KL}\bar u_{l,KL}=&0
,\qquad& 
\sum_l c_{l,KL}\bar v_{l,KL}=&0
,\label{JfQQQd}
\end{alignat}
from $(A,B,C)=(Q,Q,Q)$:
\begin{alignat}{2}
\sum_{l} \big( a_{l,KL}s_{l,NM} + a_{l,NL}s_{l,KM} \big) &= 0
,\qquad&
\sum_{l} \big( b_{l,KL}t_{l,NM} + b_{l,NL}t_{l,KM} \big) &= 0
,\nn\\
\sum_{l} \big( c_{l,KL}u_{l,NM} + c_{l,NL}u_{l,KM} \big) &= 0
,\qquad&
\sum_{l} \big( c_{l,LK}v_{l,NM} + c_{l,LN}v_{l,KM} \big) &= 0
,\nn\\
\sum_{l} \big( a_{l,KL}t_{l,NM} + c_{l,LN}u_{l,KM} \big) &= 0
,\qquad&
\sum_{l} \big( b_{l,KL}s_{l,NM} - c_{l,NL}v_{l,KM} \big) &= 0
,\label{JfQQQ}
\end{alignat}
and from $(A,B,C)=(Q,Q,B)$:
\begin{align}
 \sum_M \big( s_{l,LM}a_{m,KM}- s_{l,KM}a_{m,LM} \big) &= i\sum_k f_{klm}a_{k,KL}
,\nn\\
 \sum_M \big( u_{l,LM}c_{m,KM}- u_{l,KM}c_{m,LM} \big) &= i\sum_k h_{klm}a_{k,KL}
,\nn\\
 \sum_M \big( t_{l,LM}b_{m,KM}- t_{l,KM}b_{m,LM} \big) &= i\sum_k f_{klm}b_{k,KL}
,\nn\\
 \sum_M \big( v_{l,LM}c_{m,MK}- v_{l,KM}c_{m,ML} \big) &= -i\sum_k h_{klm}b_{k,KL}
,\nn\\
 \sum_M \big( t_{l,LM}c_{m,KM}- s_{l,KM}c_{m,ML} \big) &= -i\sum_k g_{lkm}c_{k,KL}
,\nn\\
 \sum_M \big( v_{l,LM}a_{m,KM}+ u_{l,KM}b_{m,LM} \big) &= i\sum_k g_{klm}c_{k,KL}
.\label{JfQQB}
\end{align}
From these conditions, introducing the generator $Z^{11}_{KL}:=\sum_l c_{l,ML}B^{11}_l\in \g_{(1,1)}$, it becomes the $\Z_2^2$-central charge, which means
\begin{alignat}{2}
 [Q^{01}_{\ua\,K},\, Q^{10}_{\ub\,L}]&=\varepsilon_{\ua\ub} Z^{11}_{KL},\quad &
\llbracket Z^{11}_{KL}, \,T \rrbracket &=0
,\nn\\
 [Q^{01}_{\da\,K},\, Q^{10}_{\db\,L}]&=\varepsilon_{\da\db} \bar Z^{11}_{KL},&
\llbracket \bar Z^{11}_{KL}, \,T \rrbracket &=0,
\end{alignat} 
where $T\in \{ P_\mu, Z^{(1)}_{MN}, \bar Z^{(1)}_{MN}, Z^{(2)}_{MN}, \bar Z^{(2)}_{MN}, Z^{11}_{MN}, \bar Z^{11}_{MN}, Q_{\ua\,K}^{01}, \bar Q_{\da\,K}^{01}, Q_{\ua\,K}^{10}, \bar Q_{\da\,K}^{10}\}$.
From these observations, we can obtain the novel possible generators of symmetries of the $S$-matrix:
\begin{alignat}{2}
 \g_{(0,0)}&\ni P_\mu, M_{\mu\nu}, B^{00}_l, Z^{(1)}_{KL}
, \bar Z^{(1)}_{KL}, Z^{(2)}_{KL}, \bar Z^{(2)}_{KL}
,\qquad&
 \g_{(1,1)}&\ni B^{11}_l, Z^{11}_{KL}, \bar Z^{11}_{KL}
,\nn\\
 \g_{(0,1)}&\ni Q_{\ua\,K}^{01}, \bar Q_{\da\,K}^{01}
,\qquad&
 \g_{(1,0)}&\ni Q_{\ua\,K}^{10}, \bar Q_{\da\,K}^{10}
.
\end{alignat}
Therefore the following theorem holds.
\begin{thm}
There exits the unique $\Z_2^2$-graded extension of supersymmetric algebra which is given by the $\Z_2^2$-graded LSA $\langle P_\mu, Z^{(1)}_{KL}, \bar Z^{(1)}_{KL}, Z^{(2)}_{KL}, \bar Z^{(2)}_{KL}, Z^{11}_{KL}, \bar Z^{11}_{KL}, Q_{\ua\,K}^{01}, \bar Q_{\da\,K}^{01}, Q_{\ua\,K}^{10}, \bar Q_{\da\,K}^{10}\rangle$. We call this $\mathcal{N}=(K|L)$ $\Z_2^2$-supersymmetric algebra, where $K$ and $L$ are the number of the $(0,1)$ and $(1,0)$-graded supercharges.
\end{thm}
Specifically, if we set the number of each supercharges $Q_{\ua\,K}^{01}, Q_{\ua\,L}^{10}$ to $K=L=1$, it corresponds to the $\mathcal{N}=2\ \Z_2^2$-supersymmetric algebra, which was introduced by Bruce\cite{Bruce}.  
In the present case of $\Z_2^2$-grading, the internal symmetries are generated by ordinary Lie algebra $B^{00}_l$ and $B^{11}_l$. This does not seem to be a big difference from the ordinary supersymmetry. However we will see that internal symmetries are generated by a graded Lie algebra when $\Z_2^n$-graded ($n\geq 3$) extensions of supersymmetry are considered.
\subsection{Possible $\Z_2^n$-graded generators of symmetries}
The discussion for the $\Z_2^2$-graded LSA in the previous section can also be applied to the $\Z_2^n$-graded LSA. We define two subsets of $\Z_2^n$:
\begin{align}
 (\Z_2^n)_{\ev}&:=\{\vec{b}\in \Z_2^n\  |\ \vec{b}\cdot \vec{b}=0 \mod 2\}\nn\\
 (\Z_2^n)_{\od}&:=\{\vec{f}\in \Z_2^n\  |\ \vec{f}\cdot \vec{f}=1 \mod 2\}
,
\end{align}
and the subalgebras of \eqref{Z2n-VS}:
\begin{align}
 \h_{\vec{b}}=\g_{\vec{0}}\oplus \g_{\vec{b}}\quad (\vec{b}\ne \vec{0})
,\qquad
 \h_{\vec{f}}=\g_{\vec{0}}\oplus \g_{\vec{f}}
,
\end{align}
where $\vec{0}$ is the identity element of $\Z_2^n$. 
Then $\h_{\vec{b}}$ is a Lie algebra, while $\h_{\vec{f}}$ is a LSA. This shows that we can obtain generators by applying the \CM theorem to $\h_{\vec{b}}$ and the \HLS theorem to $\h_{\vec{f}}$.  It can be seen that the following generators of the $S$-matrix are allowed:
\begin{alignat}{3}
 \g_{\vec{0}}&\ni P_\mu, M_{\mu\nu}, B^{\vec{0}}_l
,\qquad&
 \g_{\vec{b}}&\ni B^{\vec{b}}_l
,\qquad&
 \g_{\vec{f}}&\ni Q_{\ua\,K}^{\vec{f}}, \bar Q_{\da\,K}^{\vec{f}}
,
\end{alignat}
and the algebraic relations become
\begin{alignat}{2}
\llbracket Q_{\ua\,K}^{\vec{f}}, \bar Q_{\da\,L}^{\vec{f}\,'} \rrbracket &=\sigma_{\ua\da}^{\mu} P_\mu\delta_{KL}\delta^{\vec{f}\vec{f}\,'}
, \qquad&
\llbracket Q_{\ua\,K}^{\vec{f}}, Q_{\ub\,L}^{\vec{f}\,'} \rrbracket
&=\varepsilon_{\ua\ub} \sum_l a^{(\vec{f},\vec{f}\,')}_{l,KL} B^{\vec{f}+\vec{f}\,'}_l
,\nn\\
\llbracket Q^{\vec{f}}_{\ua\,L},B_l^{\vec{b}} \rrbracket 
&=s^{(\vec{f},\vec{b})}_{l,LM}Q^{\vec{b}+\vec{f}}_{\ua\,M}
, &
\llbracket B_l^{\vec{b}},B_m^{\vec{b}\,'} \rrbracket 
&=i\sum_k h^{(\vec{b},\vec{b}\,')}_{lmk} B_k^{\vec{b}+\vec{b}\,'}
,\nn\\
\llbracket \bar Q_{\da\,K}^{\vec{f}}, \bar Q_{\db\,L}^{\vec{f}\,'} \rrbracket
&=-(-1)^{(\vec{f},\vec{f}\,')}\varepsilon_{\da\db} \sum_l &\bar a^{(\vec{f},\vec{f}\,')}_{l,KL} B^{\vec{f}+\vec{f}\,'}_l
,
\label{Z2n-SUSY}
\end{alignat}
where $a,\,s$, and $h$ are structure constants with the $\vec{0}$-grading, and $a^{(\vec{f},\vec{f}\,')}_{l,KL}=(-1)^{\vec{f} \cdot \vec{f}\,' }a^{(\vec{f}\,',\vec{f})}_{l,LK}$ and $h^{(\vec{b},\vec{b}\,')}_{lmk}=-(-1)^{\vec{b}\cdot\vec{b}\,'}h^{(\vec{b}\,',\vec{b})}_{mlk}$.

In the same way as $\Z_2^2$-extension, we obtain  the conditions for structure constants from the $\Z_2^n$-Jacobi identities \eqref{Z2n-JI}. They are obtained from $(A,B,C)=(Q,\bar Q, B)$:
\begin{align}
 s_{l,LM}^{(\vec{f},\vec{f}+\vec{f}\,')}
&= \bar s_{l,ML}^{(\vec{f}\,',\vec{f}+\vec{f}\,')}
,
\end{align}
from $(A,B,C)=(Q, Q,\bar Q)$:
\begin{align}
 \sum_{l} a^{(\vec{f},\vec{f}\,')}_{l,LM} \bar s^{(\vec{f}\,'',\vec{f}+\vec{f}\,')}_{l,NK} &= 0 
,
\end{align}
from $(A,B,C)=(Q, Q, Q)$:
\begin{align}
 \sum_{l} a^{(\vec{f}\,',\vec{f}\,'')}_{l,LM} s^{(\vec{f},\vec{f}\,'+\vec{f}\,'')}_{l,KN} -(-1)^{\vec{f}\cdot\vec{f}\,'} \sum_l a^{(\vec{f},\vec{f}\,'')}_{l,KM} s^{(\vec{f}\,',\vec{f}+\vec{f}\,'')}_{l,LN} 
&= 0 
,
\end{align}
and from $(A,B,C)=(Q, Q, B)$:
\begin{align}
 s^{(\vec{f}\,',\vec{b})}_{l,LM}a^{(\vec{f},\vec{f}\,'+\vec{b})}_{m,KM}+(-1)^{\vec{f}\cdot\vec{f}\,'}s^{(\vec{f},\vec{b})}_{l,KM}a^{(\vec{f}\,',\vec{f}+\vec{b})}_{m,LM}&= \sum_{k} ih^{(\vec{f}+\vec{f}\,',\vec{b})}_{klm}a^{(\vec{f},\vec{f}\,')}_{k,KL}
.
\end{align}
Using these conditions and setting $Z_{KL}^{\vec{f}+\vec{f}\,'}:=\sum_l a^{(\vec{f},\vec{f}\,')}_{l,KL} B^{\vec{f}+\vec{f}\,'}_l$, then it satisfies
\begin{align}
 \llbracket Z_{KL}^{\vec{f}+\vec{f}\,'} ,\, T \rrbracket =
 \llbracket \bar Z_{KL}^{\vec{f}+\vec{f}\,'} ,\, T \rrbracket =0
,\quad
T\in\{P_\mu, Z_{KL}^{\vec{f}+\vec{f}\,'}, \bar Z_{KL}^{\vec{f}+\vec{f}\,'}, Q_{\ua\, K}^{\vec{f}}, \bar Q_{\ua\, K}^{\vec{f}}\},
\end{align}
where $Z_{KL}^{\vec{f}+\vec{f}\,'}=(-1)^{\vec{f}\cdot\vec{f}\,'}Z_{LK}^{\vec{f}+\vec{f}\,'}$.

Therefore the following theorem holds.
\begin{thm}
There exits the unique $\Z_2^n$-graded extension of supersymmetric algebra which is given by the $\Z_2^n$-graded LSA $\langle P_\mu, Z_{KL}^{\vec{f}+\vec{f}\,'}, \bar Z_{KL}^{\vec{f}+\vec{f}\,'}, Q_{\ua\, K}^{\vec{f}}, \bar Q_{\ua\, K}^{\vec{f}} \rangle$. 
\end{thm}
This algebra contains $K_i$ supercharges $Q_{\ua\, K_i}^{\vec{f}}$ in each subspace $\g_{\vec{f}}$ and nontrivial centers in each subspace $\mathfrak{g}_{\vec{b}}$. We call this $\mathcal{N}=(K_1| K_2| ... |K_{2^{n-1}})$ $\Z_2^n$-supersymmetric algebra.
Similarly to the $\Z_2^2$-case, if we set the number of each $Q_{\ua, K}^{\vec{f}}$ to $K_i=1$, it corresponds to the $\mathcal{N}=2^{n-1}\ \Z_2^n$-supersymmetric algebra, which was introduced by Bruce\cite{Bruce}. 
\subsection{$\Z_2^n$-Graded internal symmetries}
In this section, we focus on the generators of internal symmetry with $\Z_2^n$-grading $B_l^{\vec{b}}$. The discussion in the previous section says there can exist $B_l^{\vec{b}}$ in each subspace $\g_{\vec{b}}$ and they satisfy 
\begin{align}
 \llbracket B_l^{\vec{b}},B_m^{\vec{b}\,'} \rrbracket 
&=i\sum_k h^{(\vec{b},\vec{b}\,')}_{lmk} B_k^{\vec{b}+\vec{b}\,'}.
\tag{\ref{Z2n-SUSY}}
\end{align}
This means the set of $B_{l}^{\vec{b}}$ is a $(\Z_2^n)_{\ev}$-graded LSA. Especially, if we set $n=2k+1$, we can prove that this algebra is isomorphic to the $\Z_2^{2k}$-graded Lie algebra, which is another type of $\Z_2^n$-graded extension of Lie algebra introduced by Rittenberg and Wyler\cite{RW1,RW2}. 

\begin{dfn}\label{defLA}
Let $\g_{LA}$ be a $\Z_2^{2n}$-graded vector space \eqref{Z2n-VS} and $\vec{\ua}=(\vec{\ua}_1,\vec{\ua_2},...,\vec{\ua}_n)\in \Z_2^{2n}$ where $\vec{\ua}_i=(\ua_{2i-1},\ua_{2i})\in \Z_2^2$. Then $\g_{LA}$ is called $\Z_2^{2n}$-graded Lie algebra if $\g_{LA}$ has the $\Z_2^{2n}$-Lie bracket defined by
\begin{align}
 [A,B\}:=AB-(-1)^{\sum_{i=1}^n |\vec{\ua}_i\times \vec{\ub}_i|}BA,\quad A\in \g_{\vec{\ua}},\  B\in \g_{\vec{\ub}}
,\label{Z2n-LieB}
\end{align}
and $[A,B\}$ satisfies 
\begin{align}
 [A,B\}&\in\g_{\vec{\ua}+\vec{\ub}},\nn\\
 [A,[B,C\}\}&=[[A,B\},C\}+(-1)^{\sum_{i=1}^n |\vec{\ua}_i\times \vec{\ub}_i|}[B[A,C\}\}
\end{align}
\end{dfn}

\begin{thm}
 The $(\Z_2^{2n+1})_{\ev}$-graded LSA 
\begin{align}
\g_{\ev}:=\bigoplus_{\vec{b}\in(\Z_2^{2n+1})_{\ev}}\g_{\vec{b}}
\end{align}
 is isomorphic to the $\Z_2^{2n}$-graded Lie algebra $\g_{LA}$.
\end{thm}
\begin{proof}
Let $f$ be an injective map $f: \vec{\ua}\in \Z_2^{2n}\to\vec{a}\in(\Z_2^{2n+1})_{\ev}$, defined by
\begin{align*}
\vec{a}=\left( \vec{\ua}_1, \vec{\ua}_2+|\vec{\ua}_1|(1,1),\cdots \vec{\ua}_i+\sum_{k=1}^{i-1}|\vec{\ua}_k|(1,1),\cdots, \vec{\ua}_{n}+\sum_{k=1}^{n-1}|\vec{\ua}_k|(1,1), \sum_{k=1}^{n}|\vec{\alpha}_k| \right)
\end{align*}
where  $|\vec{a}_k|^2:=a_{2k-1}^2+a_{2k}^2$, and it is equal to $|\vec{a}_k|=a_{2k-1}+a_{2k} \mod 2$.

Since $\Z_2^{2n}$-graded Lie algebra $\g_{LA}$ and $(\Z_2^{2n+1})_{\ev}$-graded LSA $\g_{\ev}$ have the same number of the graded subspaces, the only difference of them is the Lie bracket, which comes from $\sum_{i=1}^n |\vec{\ua}_i\times \vec{\ub}_i|$ in \eqref{Z2n-LieB} and $\vec{a}\cdot \vec{b}$ in \eqref{Z2n-LB}. However, taking into account that computations are performed in modulo 2, we have straightforwardly 
\begin{align}
 \sum_{i=1}^n |\vec{\ua}_i\times \vec{\ub}_i|=\vec{a}\cdot \vec{b}.
\end{align}
Therefore the map $f$ is an isomorphism between $\g_{LA}$ and $\g_{\ev}$.
\end{proof}
From this theorem, we can redefine $\Z_2^{2n}$-graded Lie algebra (Definition \ref{defLA}).
\begin{dfn}
 Let $\g$ be a $\Z_2^{n}$-graded LSA in Definition \ref{defsalg}. We define $\Z_2^{n-1}$-graded Lie algebra as the subalgebra 
\begin{align}
 \g_{\ev}:=\bigoplus_{\vec{b}\in(\Z_2^n)_{\ev}}\g_{\vec{b}}\subset \g.
\end{align}
\end{dfn}
This is a more general definition, which means we can define not only $\Z_2^{2n}$-graded Lie algebras, but also $\Z_2^{2n+1}$-graded Lie algebras. Finally, for the internal symmetry, the following theorem holds.
\begin{thm}
 The $\Z_2^n$-graded Lie algebra is allowed as the generators of the internal symmetry of S-matrix.
\end{thm}
This is very interesting result because we can consider new types of gauge symmetries in relativistic QFT.
Specifically, $\Z_2^2$-extension of $\mathfrak{sl}_n$ has already been introduced in \cite{RW2} and some representation of $\Z_2^2$-$\mathfrak{sl}_2$ was investigated in \cite{AIKTT}. These will be helpful to construct $\Z_2^2$-$\mathfrak{sl}_n$ gauge field theory.

\section{Concluding remarks}
In this paper, we explored novel possible generators of symmetries of $S$-matrix in order to apply $\Z_2^n$-graded Lie superalgebra to relativistic quantum field theory. As a result, we found $\mathcal{N}=(K_1|K_2|\cdots |K_{2^{n-1}})$ $\Z_2^n$-graded supersymmetry is the only allowed extension of supersymmetric algebra as the $\Z_2^n$-graded extension. Specifically, if we set the number of each $Q_{\ua\, K}^{\vec{f}}$'s to $K_i=1$, they correspond to the $\mathcal{N}=2^{n-1}\ \Z_2^n$-supersymmetric algebra, which was introduced by Bruce\cite{Bruce}. However we cannot determine the number of supercharges purely from algebraic restrictions. This should be determined by the relationship between the degrees of freedom of particles within a model under consideration. 

Nontrivial results of the present paper are summarized as follows.
\begin{enumerate}
 \item It is possible to construct $\Z_2^n$-graded supersymmetric relativistic QFT.
 \item Gauge symmetry is not necessary generated by a Lie algebra, but by a $\Z_2^n$-graded Lie algebra.   
\end{enumerate} 
Thus an important future work is to construct a $\Z_2^2$-$\mathfrak{sl}_2$ gauge field theory.

Finally, in this paper, we have only considered the case of massive particles.
According to \cite{CM,HLS}, conformal generators can exist in the case of massless particles.
Therefore if we explore novel generators in the massless case, we could obtain more nontrivial $\Z_2^n$-graded conformal generators.  

\section*{Acknowledgment}
This work was supported by JST SPRING, Grant Number JPMJSP2139 (RI and AN). RI and AN would like to express our sincere gratitude to Professor Naruhiko Aizawa and Professor Nobuhito Maru, at Osaka Metropolitan University, for their invaluable guidance and support throughout our research.


\end{document}